\documentclass[journal]{IEEEtran}

\usepackage{amsmath,amssymb,amsfonts,amsthm,bm}
\usepackage{algorithm}
\usepackage{algorithmic}
\usepackage{graphicx}
\usepackage{booktabs}
\usepackage{multirow}
\usepackage[table]{xcolor}
\usepackage[caption=false,font=footnotesize]{subfig}
\usepackage{cite}
\usepackage{url}
\usepackage{hyperref}
\usepackage{afterpage}

\hypersetup{hidelinks}
\graphicspath{{figures/}}

\newcommand{\EE}{\mathbb{E}}

\newcommand{\tr}{\operatorname{tr}}
\newcommand{\relu}[1]{\left[#1\right]_{+}}

\definecolor{rmworldRed}{RGB}{178,0,0}
\definecolor{rmworldBlue}{RGB}{0,58,145}
\newcommand{\best}[1]{\textcolor{rmworldRed}{\textbf{#1}}}
\newcommand{\second}[1]{\textcolor{rmworldBlue}{\underline{#1}}}
\newtheorem{theorem}{Theorem}
\newtheorem{proposition}{Proposition}

\begin{document}

\title{RMWorld: Task-Aware Radio World Models with Value-of-Information Guided Multi-Trial Learning for Multi-UAV Communication Control}

\author{Xiucheng Wang, Nan Cheng, and Junxi Huang%
\thanks{Xiucheng Wang, Nan Cheng, and Junxi Huang are with the School of Telecommunications Engineering, Xidian University.}%
\thanks{Corresponding author: Nan Cheng.}}

\maketitle

\begin{abstract}
Reliable multi-UAV communication control depends on predicting which aerial links will serve traffic before measurements are available. Radio world models (radio WMs) make such planning tractable, but their errors are nonuniform: a globally accurate model may still fail along high-demand corridors or association boundaries where rate errors reverse control decisions. This mismatch creates a learning challenge. Link queries must reduce decision-relevant channel uncertainty, while counterfactual trials must be filtered so that biased rollouts do not corrupt the policy. Existing acquisition and model-based control treat these budgets separately, valuing uncertainty, coverage, or optimistic return rather than risk reduction. We present RMWorld, a task-aware radio-WM framework that couples value-of-information channel calibration with credibility-diversity multi-trial selection. A biased propagation formula is corrected by a Bayesian residual, and each link is valued by its exact one-label reduction in locally linearized task-integrated posterior rate variance. Counterfactual branches are selected by a task-gated log-determinant objective, followed by conflict projection and fixed-batch validation. We derive the variance-reduction identity, prove posterior task-risk equivalence and the submodular greedy guarantee, and establish a scoped first-order non-interference result. Across 100 paired 3GPP trials RMWorld reaches 0.949~bit/s/Hz task-weighted RMSE, and across 30 severe-load DeepMIMO trials it reduces median backlog by 0.967 versus Ensemble UCB at 37.5\% more offline rollouts.
\end{abstract}


\section{Introduction}
\label{sec:introduction}

Channel knowledge is the substrate of multi-UAV communication control: it determines where aerial access points should move, which users should be associated, and which traffic should be served first. Accurate measurements or ray tracing can reveal that knowledge, but neither can exhaust the locations, users, and future trajectories that a controller must consider. A radio world model (radio WM) replaces this dense interaction with a predictive channel environment learned from propagation formulas, sparse labels, or neural features. This replacement is operationally important because every planning decision can be rehearsed in the radio WM before scarce flight time, spectrum, or high-fidelity simulation is consumed. It is also dangerous because the controller learns the model's errors together with its useful structure. The central design problem is therefore not merely to fit a globally accurate channel predictor or to optimize a policy inside it. It is to decide which evidence can make an imperfect radio WM reliable for the decisions that matter. Figure~\ref{fig:rmworld-overview} summarizes the resulting design logic through coupled evidence allocation at the channel and policy levels.

\begin{figure*}[!t]
  \centering
  \includegraphics[width=0.98\textwidth]{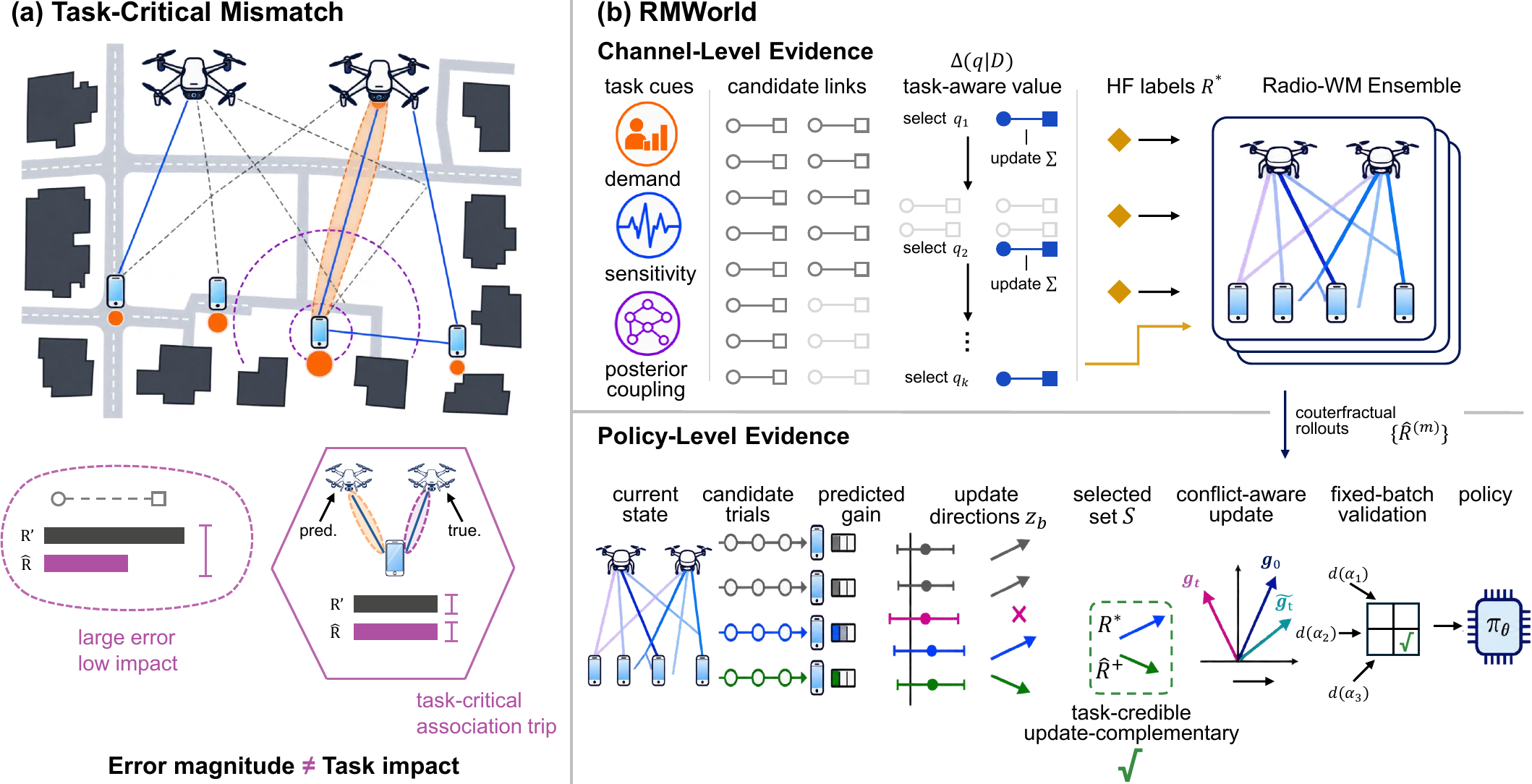}
  \caption{Motivation and workflow of RMWorld. Panel (a) illustrates why channel-error magnitude alone does not determine task impact, and panel (b) shows the coupled channel-level and policy-level evidence-allocation pipeline.}
  \label{fig:rmworld-overview}
\end{figure*}

\afterpage{%
\begin{figure}[!t]
  \centering
  \includegraphics[width=\columnwidth]{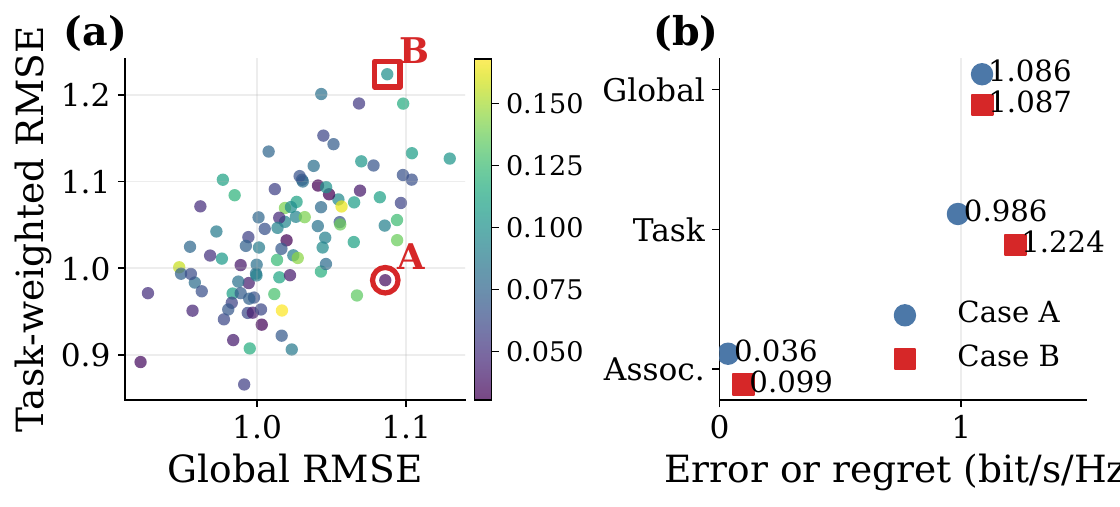}
  \caption{Global error is not task error in standardized 3GPP TR~36.777 analytic simulations. Panel (a) shows 100 frozen initial formula-radio-WM trials. Panel (b) isolates cases A and B through a fixed near-match rule on global RMSE: their global RMSEs differ by only 0.0013~bit/s/Hz, yet their task-weighted RMSEs differ by 0.2373~bit/s/Hz and their association regrets by 0.0634~bit/s/Hz. This pair is the \emph{largest} task gap among the 74 pairs that satisfy the same near-match rule, shown because it makes the mechanism legible; the population quantiles reported in Sec.~\ref{sec:3gpp-results} are the quantitative claim. These data are simulated, not measured.}
  \label{fig:motivation-diagnostic}
\end{figure}
}

The difficulty is that radio-WM fidelity is valuable only through its effect on a downstream decision. Errors near an association boundary can reverse the served user, while larger errors in regions the fleet never visits can be harmless. The same asymmetry appears along high-demand corridors, near blockage transitions, and wherever rate sensitivity changes sharply. Consequently, global reconstruction error can regard two radio WMs as equivalent even when one induces substantially higher task risk. The population comparison and matched spatial example in Fig.~\ref{fig:motivation-diagnostic} make this structural mismatch visible; the figure is evidence for the argument, not the argument itself. Three additional couplings make the problem harder. Pointwise uncertainty need not imply value of information, task relevance can repeatedly select correlated links, and channel labels operate at a different granularity from counterfactual policy trials. A sound learner must therefore allocate evidence across both levels without treating uncertainty, relevance, or diversity as interchangeable notions.

Existing solutions leave this requirement unresolved because they optimize useful but partial surrogates. Radio-map reconstruction and active surveying primarily seek map-wide fidelity, geometric coverage, or local uncertainty, while model-based controllers commonly retain trials by predicted return, action spread, or ensemble optimism. The first family can spend labels on propagation regions that do not affect service, and the second can trust a diverse but systematically biased counterfactual. Even a task-weighted uncertainty rule remains myopic when several selected links convey the same posterior information. Fig.~\ref{fig:motivation-acquisition} illustrates the missing distinction: useful queries do not merely cluster near users, but trace residual directions whose correction changes demand-weighted rate predictions across the task distribution. Sequential covariance updates then make already explained directions less attractive, converting coverage into nonredundant task evidence. This logic motivates a joint view of radio-WM calibration and multi-trial control as two budgeted evidence-selection problems. The formal gap is to quantify the reduction each candidate causes, rather than score only what the candidate currently looks like.

RMWorld fills this gap through value-of-information-guided multi-trial learning. At the channel level, it expresses the discrepancy between an estimated propagation formula and the realized formula channel as a Bayesian residual, then selects the link with the largest exact one-query reduction in locally linearized task-integrated posterior rate variance. This reduction couples task occupancy, rate sensitivity, candidate uncertainty, and cross-link covariance in one score. At the policy level, RMWorld evaluates several counterfactual trials but retains only directions that are both credible under ensemble disagreement and complementary in update space. A task-gated log-determinant objective supplies this credibility-diversity tradeoff, while conflict projection and fixed-batch validation limit harmful auxiliary updates. The analysis mirrors the two levels rather than forcing them under one claim: a posterior task-risk identity justifies channel acquisition, submodularity supports greedy branch selection, and a scoped first-order result characterizes projection. These mechanisms close the loop from task demand to channel evidence, from calibrated rate predictions to counterfactual labels, and from selected labels to the deployed controller. The central insight is that an imperfect radio WM should be improved and trusted according to the same downstream task geometry.

The contributions are fourfold.
\begin{enumerate}
  \item We formulate imperfect-radio-WM learning as two coupled evidence-allocation problems: which channel link should calibrate the radio WM and which counterfactual trial should update the controller. The formulation makes downstream task risk, rather than global map error, the common design currency.
  \item We derive an acquisition score equal to the exact one-label reduction in task-integrated posterior rate variance under a linear-Gaussian residual and local rate linearization. The derivation exposes how task geometry, rate sensitivity, and residual covariance jointly determine label value.
  \item We develop a policy-level selector whose task-gated log-determinant objective is monotone submodular, prove its cardinality-constrained greedy guarantee, and state a deliberately scoped first-order result for the unpreconditioned projected direction. A fixed-batch line search handles the finite-step case left outside that result.
  \item We provide a reproducible evidence chain across 100 paired 3GPP formula trials, 30 paired severe-load DeepMIMO trials, 30 load-boundary probes, and a paired neural-feature study. RMWorld attains 0.949~bit/s/Hz task-weighted RMSE in the formula study and reduces severe-load median backlog by 0.967 versus Ensemble UCB, while the audits expose unresolved decision endpoints and the 37.5\% rollout overhead.
\end{enumerate}

\afterpage{%
\begin{figure}[!t]
  \centering
  \includegraphics[width=\columnwidth]{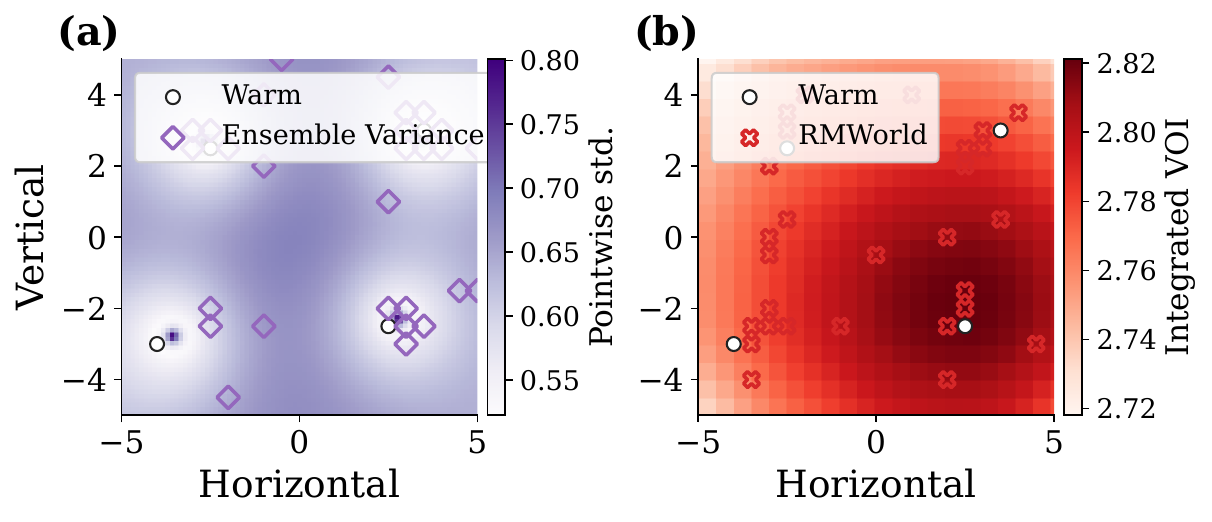}
  \caption{Query allocation on a representative frozen 3GPP trial. Panel (a) shows that ensemble variance follows locally uncertain links. Panel (b) shows that RMWorld values how one label reduces integrated rate uncertainty over the task distribution and updates its covariance after each selection. Both methods receive four warm-start labels and the same total budget of 32 labels.}
  \label{fig:motivation-acquisition}
\end{figure}
}

\section{Related Work}
\label{sec:related-work}

\subsection{Radio World Models for UAV Communication Control}

UAV-assisted communication couples trajectory planning, user association, and resource allocation through a spatially heterogeneous channel~\cite{zeng2019accessing,mozaffari2019tutorial}. Reinforcement learning addresses repeated movement and interference decisions that resist direct optimization~\cite{challita2019interference}, while simultaneous navigation and radio mapping makes channel acquisition part of the motion policy~\cite{zeng2021snarm}. Model-based policy search reduces physical interaction through probabilistic dynamics, ensembles, and short-horizon rollouts~\cite{deisenroth2011pilco,chua2018pets,janner2019mbpo}. These methods establish the value of learning an environment model, but their control guarantee remains conditional on where that model is accurate. Radio maps and channel knowledge maps (CKMs) provide the communication-specific representation by encoding location-dependent propagation structure~\cite{chen2017radiomaps,zeng2021ckm,zeng2024ckmtutorial}. Neural estimators reconstruct dense maps from geometry or sparse observations, including RadioUNet, completion autoencoders, and diffusion-based RadioDiff~\cite{levie2021radiounet,teganya2022completion,wang2025radiodiff}. Active UAV surveying further chooses trajectories that reduce mapping uncertainty~\cite{shrestha2023surveying}, and environment CKMs have been consumed by joint UAV trajectory and communication decisions~\cite{zhan2024eckm}. This literature makes radio environments learnable and actionable, yet it usually optimizes reconstruction, coverage, survey efficiency, or control with a given map rather than the downstream value of each newly acquired link.

\subsection{Task-Aware Acquisition and Multi-Trial Learning}

Ensembles provide competitive epistemic signals for conservative prediction and model-based control~\cite{kurutach2018modelensemble,ciosek2020betaqo}, which motivates variance and optimistic-confidence baselines. However, pointwise variance describes uncertainty at a candidate rather than the task-risk reduction caused by observing it, while task weighting alone can repeatedly select correlated links. Log-determinant experimental design addresses redundancy by rewarding complementary directions~\cite{krause2008near}, and cardinality-constrained greedy maximization has the classical approximation guarantee for normalized monotone submodular objectives~\cite{nemhauser1978analysis}. RMWorld uses these principles selectively rather than conflating them: channel acquisition is an integrated A-optimal variance reduction from a rank-one posterior update, whereas policy-trial selection uses a task-gated log-determinant objective. Counterfactual model rollouts can expose several plausible updates from one nominal state~\cite{chua2018pets,kurutach2018modelensemble}, but optimistic return alone is brittle under radio-WM bias and diversity alone can preserve reliably poor trials. The proposed selector therefore gates predicted improvement by ensemble disagreement and represents each trial by the gradient of the update it would actually induce. Conflict projection follows the broader principle of protecting a nominal objective from competing gradients, while fixed-batch validation is related to trust-region control of finite updates~\cite{yu2020pcgrad,schulman2015trust}. Unlike prior work that optimizes radio-map acquisition or multi-trial policy search in isolation, RMWorld connects them through downstream task geometry and explicitly separates the posterior guarantees from realized error under model misspecification.

\section{System Model and Problem Formulation}
\label{sec:system}

\subsection{Multi-UAV Communication Task}

Consider $N$ UAVs serving $U$ ground users over a planar region $\mathcal X$. UAV $i$ has position $\bm p_i^t\in\mathcal X$ at decision step $t$, and user $u$ has position $\bm q_u$ and remaining traffic $Q_u^t$. A shared policy $\pi_{\bm\theta}$ maps UAV positions, user locations, remaining demands, and geometric context to movement actions. Let $R_{iu}^{\star}(\bm p_i^t)$ denote the spectral efficiency under the unknown channel and let $a_{iu}^t\in\{0,1\}$ denote the scheduler assignment. The queue evolves according to
\begin{align}
  Q_u^{t+1}
  &=\relu{Q_u^t-\tau\sum_{i=1}^{N}a_{iu}^t
  R_{iu}^{\star}(\bm p_i^t)},
  \label{eq:queue-dynamics}
\end{align}
where $\tau$ absorbs the resource-block duration and bandwidth normalization. The control endpoint is the normalized cumulative backlog area
\begin{align}
  L(\bm\theta;R^{\star})
  &=\frac{1}{T\sum_u Q_u^0}
  \sum_{t=1}^{T}\sum_{u=1}^{U}Q_u^t,
  \label{eq:control-loss}
\end{align}
which penalizes both unfinished traffic and late service.

The policy is trained without dense access to $R^{\star}$. Instead, an ensemble of radio WMs $\{\widehat R^{(m)}\}_{m=1}^{M}$ supplies differentiable rollout rates. High-fidelity access reveals $R^{\star}_{iu}(\bm p)$ only for queried links, which may correspond to measurement, a detailed simulator, or held-out ray tracing. The budget is separated into link labels for radio-WM calibration and audited counterfactuals for policy training. This separation is essential because a branch can expose many links while only a small subset can be labeled.

\subsection{Formula Radio World Model and Task Error}

The primary evidence uses the UMa-AV model in 3GPP TR~36.777~\cite{3GPPTR36777}. For aerial height $h$ in meters, carrier frequency $f_c$ in GHz, and two- and three-dimensional distances $d$ and $d_{3\mathrm D}$ in meters, define
\begin{align}
 d_1(h)&=\max\{460\log_{10}h-700,18\},\notag\\
 p_1(h)&=\max\{4300\log_{10}h-3800,1\},
 \label{eq:3gpp-breakpoints}
\end{align}
and
\begin{align}
 P_{\mathrm L}(d,h)
 &=\begin{cases}
 1, & d\le d_1(h),\\
 \dfrac{d_1(h)}{d}+e^{-d/p_1(h)}\!\left(1-\dfrac{d_1(h)}{d}\right),
 & d>d_1(h).
 \end{cases}
 \label{eq:3gpp-los}
\end{align}
The LOS and NLOS path losses are
\begin{align}
 \mathrm{PL}_{\mathrm L}
 &=28+22\log_{10}d_{3\mathrm D}+20\log_{10}f_c,\notag\\
 \widetilde{\mathrm{PL}}_{\mathrm N}
 &=-17.5+(46-7\log_{10}h)\log_{10}d_{3\mathrm D}\notag\\
 &\quad+20\log_{10}\!\left(\frac{40\pi f_c}{3}\right),\notag\\
 \mathrm{PL}_{\mathrm N}
 &=\max\{\mathrm{PL}_{\mathrm L},\widetilde{\mathrm{PL}}_{\mathrm N}\}.
 \label{eq:3gpp-pathloss}
\end{align}
For state $s\in\{\mathrm L,\mathrm N\}$, $\gamma_s=10^{(P_{\rm tx}-\mathrm{PL}_s-P_{\rm n})/10}$ and $R_s=\log_2(1+\gamma_s)$. The expected rate is $P_{\mathrm L}R_{\mathrm L}+(1-P_{\mathrm L})R_{\mathrm N}$. The realized formula side uses the correct physical parameters and hidden spatially correlated, state-specific shadow fields. Each estimated formula member uses biased height, frequency, LOS prior, and path-loss offsets, followed by a sparse user-specific log-rate residual.

More explicitly, let $\bm\vartheta^{\star}$ collect the correct height, carrier, and standardized LOS/NLOS parameters, and let $Z_{s,u}(x)$ be a zero-mean, unit-variance spatial field that is hidden from the learner. The realized rate oracle is
\begin{align}
 \Gamma_{s,u,x}^{\star}
 &=10^{\{P_{\rm tx}-\mathrm{PL}_s(x;\bm\vartheta^{\star})
 -\sigma_s Z_{s,u}(x)-P_{\rm n}\}/10},\notag\\
 R_{u,x}^{\star}
 &=\sum_{s\in\{\mathrm L,\mathrm N\}}p_s(x;\bm\vartheta^{\star})
 \log_2(1+\Gamma_{s,u,x}^{\star}),
 \label{eq:true-formula-rate}
\end{align}
where $p_{\mathrm N}=1-p_{\mathrm L}$. Formula radio-WM member $m$ instead predicts
\begin{align}
 \widehat R^{(m)}_{u,x}(\mathcal D)
 &=R_{\rm 3GPP}(x;\bm\vartheta_m)
 \exp\!\left(\bm\phi_x^{\mathsf T}\widehat{\bm\beta}_{m,u}(\mathcal D)\right),
 \label{eq:estimated-formula-rate}
\end{align}
where $\bm\vartheta_m$ contains biased physical parameters, $\bm\phi_x$ is a normalized radial-basis design, and $\widehat{\bm\beta}_{m,u}$ is fit only from queried rates. Thus the radio-WM error used throughout the primary evidence is literally
\begin{align}
 e_{u,x}(\mathcal D)
 &=M^{-1}\sum_{m=1}^{M}\widehat R^{(m)}_{u,x}(\mathcal D)
 -R_{u,x}^{\star}.
 \label{eq:formula-error}
\end{align}
This construction is deliberately harder than estimating one scalar path-loss exponent. It combines parametric mismatch, nonlinear propagation through SNR, a spatially structured omitted field, and sparse correction, while retaining a transparent ground-truth formula on both sides. It therefore isolates whether evidence allocation can repair a consequential formula error without attributing the result to an opaque learned oracle.

For query set $\mathcal D$, the formula radio-WM error is the difference between ensemble-mean predicted and realized formula rates at the same geometry. Let $c_u$ denote user $u$'s demand, $d_u(x)$ the distance from $x$ to that user, and $d_{\rm cor,u}(x)$ the distance to the nearer line segment joining an initial UAV location to that user. The frozen task distribution is
\begin{align}
 w_{u,x}\propto c_u\bigg[0.12
 &+0.48\exp\!\left(-\frac{d_u^2(x)}{2(2.1)^2}\right)\notag\\
 &+0.40\exp\!\left(-\frac{d_{\rm cor,u}^2(x)}{2(0.85)^2}\right)\bigg].
 \label{eq:task-weight}
\end{align}
The nonzero background term prevents the learner from declaring the rest of the map irrelevant; the neighborhood and corridor terms encode where association and service are likely to consume rates. The primary channel endpoint is
\begin{align}
 E_w(\mathcal D)
 &=\left(
 \frac{\sum_{u,x}w_{u,x}
 \left(\overline R_{u,x}(\mathcal D)-R_{u,x}^{\star}\right)^2}
 {\sum_{u,x}w_{u,x}}
 \right)^{1/2}.
 \label{eq:weighted-rmse}
\end{align}
Global RMSE assigns all evaluated links equal weight. Association regret instead compares the realized utility of the user selected by $\overline R$ with that of the user selected by $R^{\star}$. These endpoints separate model fit, task-local fit, and a discontinuous downstream decision.

\subsection{Two-Level Budgeted Learning Problem}

At training episode $e$, the policy generates a nominal branch and $B$ counterfactual branches by perturbing the common action-noise realization. The ensemble predicts each branch loss, and the learner retains at most $K$ branches for local supervision. Within the high-fidelity audit, the learner can request only a limited set of link labels. The channel decision chooses links that reduce $E_w$, while the policy decision chooses branches that improve the nominal policy without redundant gradients. The overall problem can be written as
\begin{align}
 \min_{\bm\theta,\{\mathcal D_e,S_e\}}
 &\quad \EE\!\left[L(\bm\theta;R^{\star})\right]
 +\mu\,\EE\!\left[E_w(\mathcal D_E)\right] \notag\\
 \mathrm{s.t.}
 &\quad |\mathcal D_E|\le Q,\quad |S_e|\le K,
 \label{eq:joint-problem}
\end{align}
where $Q$ is the total link-label budget and $\mu$ states the calibration role. RMWorld solves the two discrete allocations with different objectives and guarantees. It does not claim that optimizing the local surrogates globally solves~\eqref{eq:joint-problem}.

\section{RMWorld: Task-Aware Radio-WM and Multi-Trial Learning}
\label{sec:method}

\subsection{Task-Integrated Channel Evidence}

For user $u$, let $\bm\beta_u$ be the coefficient of a log-rate residual with regularized posterior covariance $\bm\Sigma_u\succeq0$. Candidate query $q=(u_q,x_q)$ has residual feature $\bm x_q$ and observation variance $\sigma_q^2$. The rate derivative at evaluation link $(u,x)$ is $\bm a_{u,x}=\nabla_{\bm\beta_u}\widehat R_{u,x}$, which equals $\widehat R_{u,x}\bm\phi_x$ for the multiplicative residual used in the formula model. With joint task mass $W=\sum_{v,x}w_{v,x}$, RMWorld forms one block per user:
\begin{align}
 \bm H_u
 &=\frac{1}{W}\sum_x w_{u,x}\bm a_{u,x}\bm a_{u,x}^{\mathsf T}.
 \label{eq:integrated-gram}
\end{align}
The value assigned to query $q$ is
\begin{align}
 \Delta(q\mid\mathcal D)
 &=\frac{\bm x_q^{\mathsf T}\bm\Sigma_{u_q}\bm H_{u_q}
 \bm\Sigma_{u_q}\bm x_q}
 {\sigma_q^2+\bm x_q^{\mathsf T}\bm\Sigma_{u_q}\bm x_q}.
 \label{eq:channel-voi}
\end{align}
The common denominator $W$ is essential: it preserves relative demand when links from different users compete for one label budget. The numerator couples current coefficient uncertainty, joint task sensitivity, and the candidate feature. The denominator discounts observations already explained by the posterior design. After each selection, RMWorld updates the affected covariance block before rescoring all candidates, making a batch responsive to information already acquired.

\begin{theorem}[Exact local variance reduction]
\label{thm:voi}
Assume independent user residuals $y_q=\bm x_q^{\mathsf T}\bm\beta_{u_q}+\eta_q$ with $\eta_q\sim\mathcal N(0,\sigma_q^2)$, covariances $\bm\Sigma_u\succeq0$, and local rate linearizations represented by $\bm H_u\succeq0$. If $V(\mathcal D)=\sum_u\tr(\bm H_u\bm\Sigma_u)$, then one observation at $q$ reduces $V$ by exactly~\eqref{eq:channel-voi}, and $\Delta(q\mid\mathcal D)\ge0$.
\end{theorem}

\begin{proof}
The Gaussian rank-one covariance update is
\begin{align}
 \bm\Sigma_{u_q}^+
 &=\bm\Sigma_{u_q}-
 \frac{\bm\Sigma_{u_q}\bm x_q\bm x_q^{\mathsf T}\bm\Sigma_{u_q}}
 {\sigma_q^2+\bm x_q^{\mathsf T}\bm\Sigma_{u_q}\bm x_q}.
 \label{eq:covariance-update}
\end{align}
Only block $u_q$ changes. Substituting~\eqref{eq:covariance-update} into $V(\mathcal D)-V(\mathcal D\cup\{q\})$ and applying trace cyclicity yields~\eqref{eq:channel-voi}. Its denominator is positive and its numerator is nonnegative because $\bm H_{u_q}$ and $\bm\Sigma_{u_q}$ are positive semidefinite.
\end{proof}

The next result connects this algebraic quantity to the experimental endpoint rather than leaving $\sum_u\tr(\bm H_u\bm\Sigma_u)$ as an abstract design criterion.

\begin{proposition}[Posterior task-risk equivalence]
\label{prop:risk-equivalence}
Let $\rho_w(u,x)=w_{u,x}/W$, let $\delta R_{u,x}=\bm a_{u,x}^{\mathsf T}(\bm\beta_u-\widehat{\bm\beta}_u)$, and assume $\bm\beta_u\mid\mathcal D$ has mean $\widehat{\bm\beta}_u$ and covariance $\bm\Sigma_u$. Then
\begin{align}
 \EE_{(u,x)\sim\rho_w}\EE\!\left[\delta R_{u,x}^2\mid\mathcal D\right]
 &=\sum_u\tr(\bm H_u\bm\Sigma_u)=V(\mathcal D).
 \label{eq:posterior-task-risk}
\end{align}
Consequently, $\Delta(q\mid\mathcal D)$ is the exact reduction in posterior expected squared task-local rate error under the stated linearization.
\end{proposition}

\begin{proof}
Conditioned on $\mathcal D$, the zero-mean coefficient error gives
$\EE[\delta R_{u,x}^2\mid\mathcal D]=\bm a_{u,x}^{\mathsf T}\bm\Sigma_u\bm a_{u,x}$. Averaging with the joint distribution $\rho_w$, applying trace cyclicity, and substituting~\eqref{eq:integrated-gram} yields~\eqref{eq:posterior-task-risk}. The final statement follows from Theorem~\ref{thm:voi}.
\end{proof}

Proposition~\ref{prop:risk-equivalence} explains why the acquisition criterion is aligned with the squared quantity inside~\eqref{eq:weighted-rmse}; the reported RMSE is its square root evaluated against the realized formula oracle. It also separates RMWorld from three plausible shortcuts. Pointwise variance estimates uncertainty at the candidate itself, whereas~\eqref{eq:channel-voi} asks how that candidate changes risk over all task links. Multiplying pointwise variance by $w_q$ recognizes an important location but ignores cross-location covariance. Determinant design rewards spanning features but need not prioritize directions that have large mass in the joint task Gram blocks. RMWorld contains all three ingredients---candidate uncertainty, task sensitivity, and posterior coupling---in one reduction identity.

Theorems~\ref{thm:voi} and Proposition~\ref{prop:risk-equivalence} are one-query statements for a linear-Gaussian residual and a local rate approximation. Within one query batch, RMWorld freezes $\{\bm H_u\}$ at the pre-batch linearization and applies the identity repeatedly while updating the affected covariance block after every label. After residual refitting, it rebuilds the rate derivatives and Gram blocks for the next batch. We neither assert global submodularity nor attach an approximation ratio to this greedy A-optimal batch. This scope distinction is important because the policy-level log-determinant objective below has a different structure and supports a different guarantee.

\subsection{Task-Gated Multi-Trial Selection}

Let $L_{b,m}$ be the cumulative backlog predicted for counterfactual branch $b$ by ensemble member $m$, with $b=0$ denoting the nominal branch. The predicted improvement is $d_{b,m}=L_{0,m}-L_{b,m}$. RMWorld uses a normalized pessimistic estimate
\begin{align}
 \underline d_b
 &=\frac{\relu{\operatorname{mean}_m(d_{b,m})-
 \beta\operatorname{std}_m(d_{b,m})}}
 {\max\{\operatorname{median}_j|\operatorname{mean}_m(d_{j,m})|,
 \epsilon_d\}},\notag\\
 u_b&=\min\{\underline d_b,1\}.
 \label{eq:task-lcb}
\end{align}
A branch contributes task utility only when its predicted improvement remains positive after ensemble disagreement is charged. This gate avoids spending a selection slot on a direction that is diverse but not credibly useful.

The branch descriptor is aligned with the update that training will actually apply. Let $q_b^+$ and $q_b^-$ be soft probabilities for positive and negative branch labels, and let $\ell_b^+$ and $\ell_b^-$ be the corresponding local imitation and rejection losses. The induced descriptor and its reliability-weighted normalization are
\begin{align}
 \bm g_b
 &=\nabla_{\bm\theta}\left(q_b^+w_+\ell_b^+
 +q_b^-w_-\ell_b^-\right),\notag\\
 \bm z_b
 &=\sqrt{r_b}\frac{\bm g_b}
 {\max\{\|\bm g_b\|_2,\epsilon_g\}},
 \label{eq:gradient-descriptor}
\end{align}
where $r_b\in[r_{\min},1]$ summarizes the magnitude and sign consistency of ensemble comparisons with the nominal branch. For selected set $S$, RMWorld maximizes
\begin{align}
 F(S)
 &=\log\det\!\left(\bm I+
 \sum_{b\in S}\bm z_b\bm z_b^{\mathsf T}\right)
 +\lambda\sum_{b\in S}u_b,
 \quad |S|\le K.
 \label{eq:set-objective}
\end{align}
The first term rewards complementary update directions, while the second preserves task credibility. Sylvester's identity permits evaluation through the smaller selected-branch Gram matrix.

\begin{theorem}[Greedy branch guarantee]
\label{thm:submodular}
Assume that $\{\bm z_b,u_b\}_{b=1}^{B}$ are fixed within one selection step, $u_b\ge0$, and $F(\varnothing)=0$. Then~\eqref{eq:set-objective} is normalized, monotone, and submodular. The cardinality-$K$ greedy set $S_{\rm g}$ satisfies
\begin{align}
 F(S_{\rm g})
 &\ge(1-e^{-1})\max_{|S|\le K}F(S).
 \label{eq:greedy-guarantee}
\end{align}
\end{theorem}

\begin{proof}
For $\bm A_S=\bm I+\sum_{j\in S}\bm z_j\bm z_j^{\mathsf T}$, the matrix determinant lemma gives marginal information gain
\begin{align}
 \Delta_b^{\rm info}(S)
 &=\log\!\left(1+\bm z_b^{\mathsf T}\bm A_S^{-1}\bm z_b\right)\ge0.
\end{align}
If $S\subseteq T$, then $\bm A_S\preceq\bm A_T$ and $\bm A_S^{-1}\succeq\bm A_T^{-1}$, which establishes diminishing returns. The task term is nonnegative and modular, so it preserves normalization, monotonicity, and submodularity. The result in~\eqref{eq:greedy-guarantee} follows from cardinality-constrained greedy maximization~\cite{nemhauser1978analysis}.
\end{proof}

\subsection{Conflict Projection and Fixed-Batch Validation}

Let $\bm g_0=\nabla_{\bm\theta}L_0$ be the nominal gradient and $\bm g_t$ the aggregate selected-branch gradient. Motivated by gradient-surgery methods for competing objectives~\cite{yu2020pcgrad}, RMWorld removes only the component of $\bm g_t$ that conflicts with the nominal direction:
\begin{align}
 \widetilde{\bm g}_t
 &=\bm g_t-
 \frac{\min\{\langle\bm g_0,\bm g_t\rangle,0\}}
 {\max\{\|\bm g_0\|_2^2,\epsilon_g\}}\bm g_0,\notag\\
 \bm d(\alpha)
 &=\frac{\bm g_0}{\max\{\|\bm g_0\|_2,\epsilon_g\}}
 +\alpha\frac{\widetilde{\bm g}_t}
 {\max\{\|\widetilde{\bm g}_t\|_2,\epsilon_g\}},\notag\\
 \overline{\bm d}(\alpha)
 &=\frac{\bm d(\alpha)}{\max\{\|\bm d(\alpha)\|_2,1\}}.
 \label{eq:conflict-update}
\end{align}
This operation concerns optimization interference, not geometric collision avoidance. It cannot certify physical safety, and the experiments report collision diagnostics separately.

\begin{proposition}[First-order non-interference]
\label{prop:noninterference}
Suppose $\|\bm g_0\|_2^2\ge\epsilon_g$. Then $\langle\bm g_0,\widetilde{\bm g}_t\rangle\ge0$. Consequently, $\langle\bm g_0,\overline{\bm d}(\alpha)\rangle\ge\|\bm g_0\|_2/(1+\alpha)$ for $\alpha\ge0$. If $L_0$ is $L$-smooth, the unpreconditioned step $\bm\theta^+=\bm\theta-\eta\overline{\bm d}(\alpha)$ decreases $L_0$ whenever
\begin{align}
 0<\eta<\frac{2\|\bm g_0\|_2}{L(1+\alpha)}.
 \label{eq:step-bound}
\end{align}
\end{proposition}

\begin{proof}
Under the stated norm condition, the numerical clamp in~\eqref{eq:conflict-update} is inactive. When $\langle\bm g_0,\bm g_t\rangle<0$, substitution makes the projected inner product zero; otherwise the projection leaves $\bm g_t$ unchanged. The triangle inequality gives $\|\bm d(\alpha)\|_2\le1+\alpha$, which yields the alignment lower bound after clipping. Applying $L$-smoothness to $-\eta\overline{\bm d}(\alpha)$ gives
\begin{align}
 L_0(\bm\theta^+)
 &\le L_0(\bm\theta)
 -\frac{\eta\|\bm g_0\|_2}{1+\alpha}+\frac{L\eta^2}{2}.
\end{align}
The right-hand side is strictly smaller under~\eqref{eq:step-bound}.
\end{proof}

The proposition applies before AdamW preconditioning and excludes the clamp-active near-zero-gradient regime, so RMWorld does not transfer this descent bound to every implemented step. Instead, it evaluates $\alpha$ in the finite set $\mathcal A=\{0,0.25\alpha_0,0.5\alpha_0,\alpha_0\}$ on a fixed common-random-number radio-WM batch. It returns zero when the best measured gain is below a margin and otherwise returns the minimizing candidate. The chosen update therefore cannot exceed the nominal update's loss on that reused validation batch. This is a finite-batch dominance statement, not a guarantee for the unknown physical channel.

\begin{algorithm}[!t]
\caption{Task-aware radio-WM and multi-trial RMWorld episode.}
\label{alg:rmworld}
\begin{algorithmic}[1]
\REQUIRE Policy $\pi_{\bm\theta}$, radio-WM ensemble, branch budget $K$, link budget $Q_e$.
\STATE Roll out the nominal branch and $B$ counterfactual branches with common random numbers.
\STATE Compute $u_b$, $r_b$, and update-aligned descriptors $\bm z_b$.
\STATE Greedily select $S_e$ by the marginal gain of~\eqref{eq:set-objective}.
\IF{link calibration is enabled}
  \STATE Expose candidate links from the audited counterfactual and build $\{\bm H_u\}$.
  \FOR{$q=1$ to $Q_e$}
    \STATE Maximize~\eqref{eq:channel-voi}, obtain the rate, and apply~\eqref{eq:covariance-update}.
  \ENDFOR
  \STATE Refit the residual radio WM and relabel selected branches.
\ENDIF
\STATE Form~\eqref{eq:conflict-update} and choose $\alpha$ on the fixed validation batch.
\STATE Update $\bm\theta$ and retain the best radio-WM checkpoint.
\end{algorithmic}
\end{algorithm}

For $B$ branches, $M$ models, rollout horizon $H_r$, and descriptor dimension $D$, pool construction costs $O(BMH_r)$ model steps and descriptor construction uses $B$ backward passes. Greedy policy selection costs $O(KBD)$ inner-product work with rank-one marginal updates. Channel acquisition evaluates candidate quadratic forms and updates a small residual covariance after each chosen link. These operations occur during offline training; deployment executes the same policy-forward structure as the single-trial reference.

\section{Experimental Evaluation}
\label{sec:experiments}
\label{sec:results}

\subsection{Evidence Sources and Protocols}

The evaluation is organized by claim rather than by dataset prestige. The 3GPP study isolates the causal channel-calibration mechanism because both the true and estimated formulas are explicit. The DeepMIMO study tests policy-level end-to-end control on held-out ray-tracing rate maps outside the formula family~\cite{alkhateeb2019deepmimo}. To avoid attributing a compound change to either level, the former holds the residual estimator fixed and varies the link selector, whereas the latter freezes the high-fidelity audit machinery and varies the branch selector. Algorithm~\ref{alg:rmworld} shows their composable loop; no DeepMIMO comparison is presented as a test of the channel-acquisition module. The neural-feature study changes the residual representation while preserving paired channel and task instances. A long-horizon stress test is retained as boundary evidence because longer training does not separate the methods reliably.

\begin{table}[!t]
\centering
\caption{Experimental protocols and inferential roles. The 3GPP blocks use TR~36.777 UMa-AV formula channels; the control blocks use the held-out DeepMIMO O1 slice. All scenes contain two UAVs and four users.}
\label{tab:protocols}
\resizebox{\columnwidth}{!}{%
\begin{tabular}{@{}l|c|c|l@{}}
\hline
Evidence block & Paired $n$ & HF budget & Endpoint / role \\
\hline
\rowcolor{blue!8} \best{3GPP formula radio WM} & \best{100} & \best{32 links} & wRMSE / primary mechanism \\
DeepMIMO control & 30; 10/load & 6 audits & Backlog / external + boundary \\
Neural-feature radio WM & 30 & 32 links & wRMSE / sensitivity \\
Long-horizon stress & 30 & 12 updates & Backlog / scope boundary \\
\hline
\end{tabular}}
\end{table}

The 3GPP experiment uses an $81\times81$ evaluation grid, which yields 26,244 user-location links, and a $21\times21$ query lattice per user, which yields 1,764 candidates. Every method receives four warm-start labels and 28 adaptive labels. The physical configuration uses a 60~m aerial height, 25~m base-station height, 3.5~GHz carrier, 30~dBm transmit power, 20~MHz bandwidth, and 7~dB noise figure. Three formula members start from different biased physical parameters and learn user-specific residual heads. Paired methods share the realized channel fields, task weights, formula ensemble, and query budget.

The DeepMIMO benchmark uses the O1 outdoor map at a held-out 42.4~m aerial slice. Four user positions, two UAVs, and two circular obstacles define each scene. Initial tasks are sampled as $(2+1.5\xi)\times8$ with $\xi\sim\mathcal U[0,1]$, placing the 30-seed main comparison in the severe-load regime selected before the formal seeds. Each seed uses six training episodes of 35 steps and a deterministic evaluation horizon of 160 steps. Multi-trial methods share a pool of five branches, retain two, use an ensemble of three radio WMs, and receive one high-fidelity audit per episode. A same-protocol boundary scan changes only the offered-load multiplier to 4, 6, or 10 and uses 10 paired seeds per load; these 30 additional runs are exploratory and are not added to the main Holm family.

\subsection{Baselines, Endpoints, and Statistics}

No published method jointly exposes the same link-label budget, counterfactual-trial budget, and imperfect-radio-WM interface, so labeling one imported algorithm as a direct end-to-end SOTA comparator would be misleading. We instead construct a category-complete state-of-the-art (SOTA) roster from the strongest contemporary principles in active radio surveying, Bayesian experimental design, ensemble model-based control, optimistic exploration, and gradient-space diversity. Each selector is reimplemented behind the same candidate interface, which makes the experiment a controlled comparison of evidence-selection principles rather than software stacks. The baselines are organized as follows.

\begin{itemize}
  \item \textbf{Channel-calibration selectors:} Random is the assumption-light floor. Spatial D-opt instantiates determinant-based sensor placement~\cite{krause2008near}; Ensemble Variance instantiates conservative epistemic acquisition~\cite{ciosek2020betaqo}; Task-Weighted Variance adds the task distribution to the uncertainty logic used in active radio surveying~\cite{shrestha2023surveying,zeng2024ckmtutorial}; and Gradient D-opt applies determinant design to rate-gradient features~\cite{krause2008near}. Together they span geometry, pointwise uncertainty, task-local uncertainty, and gradient diversity, which are the leading SOTA acquisition principles relevant to a sparse radio WM. RMWorld differs only in valuing the integrated posterior risk reduction caused by the label.
  \item \textbf{Multi-trial control selectors:} Single-Trial radio WM represents classical one-rollout model-based policy learning~\cite{deisenroth2011pilco,janner2019mbpo}. Random Multi-Trial tests the gain from a larger stochastic candidate pool; Action Entropy prioritizes behavioral spread; Ensemble UCB combines predicted return with epistemic optimism, following ensemble model-based control~\cite{chua2018pets,kurutach2018modelensemble,ciosek2020betaqo}; and Gradient Information selects diverse update directions using the gradient-space principle underlying experimental design and conflict-aware learning~\cite{krause2008near,yu2020pcgrad}. This group covers the strongest plausible alternatives based on trial count, action diversity, optimism, and update diversity. RMWorld adds task credibility to diversity rather than introducing a different controller backbone.
  \item \textbf{Fairness and inferential scope:} All multi-trial methods share the candidate pool, common random numbers, local positive/negative update, conflict projection, high-fidelity budget, and checkpoint rule. They differ in branch selection, except that RMWorld also invokes the finite radio-WM line search whose extra calls are reported explicitly. Single-Trial radio WM is an unequal-compute, low-cost reference and is not treated as an equal-budget selector. Primary conclusions therefore emphasize prespecified paired comparisons with the external multi-trial selectors and retain unresolved comparisons instead of declaring blanket SOTA superiority.
\end{itemize}

\begin{table}[!t]
\centering
\caption{Controlled baseline operators. Every selector uses the shared candidate pool and the budget of its evidence level; only the frozen ranking and batch-construction rules change.}
\label{tab:baseline-operators}
\scriptsize
\setlength{\tabcolsep}{2.2pt}
\resizebox{\columnwidth}{!}{%
\begin{tabular}{@{}l|l|l|l@{}}
\hline
Level & Selector & Frozen ranking operator & Batch rule \\
\hline
\multirow{6}{*}{Channel}
 & Random & Uniform sampling & Without replacement \\
 & Spatial D-opt & Spatial-feature log-det & Greedy \\
 & Ensemble Variance & Predictive standard deviation & Static top-$B$ \\
 & Task-Weighted Variance & Task weight $\times$ standard deviation & Static top-$B$ \\
 & Gradient D-opt & Rate-gradient log-det & Greedy \\
\rowcolor{blue!8} & \best{RMWorld} & \best{Integrated posterior-risk reduction} & \best{Sequential conditioning} \\
\hline
\multirow{6}{*}{Policy}
 & Single Trial & Nominal branch & One branch \\
 & Random Multi-Trial & Uniform sampling & Fixed-$K$ subset \\
 & Action Entropy & Action-descriptor log-det & Greedy \\
 & Ensemble UCB & Mean gain $+$ disagreement & Static top-$K$ \\
 & Gradient Information & Reliability-weighted gradient log-det & Greedy \\
\rowcolor{blue!8} & \best{RMWorld} & \best{Task-gated gradient log-det} & \best{Greedy $+$ batch validation} \\
\hline
\end{tabular}}
\end{table}

Table~\ref{tab:baseline-operators} exposes the controlled operator changed by each comparison, without assigning qualitative capabilities to a method. On the channel side, Ensemble Variance ranks candidates by predictive standard deviation, Task-Weighted Variance multiplies that score by the frozen task weight, and the two D-opt variants greedily maximize spatial- or rate-feature log-determinants. RMWorld instead evaluates how a candidate changes integrated posterior risk after conditioning, so the value-of-information score automatically discounts evidence made redundant by earlier queries. On the policy side, Action Entropy applies the same corrected identity-regularized log-determinant operator to action descriptors, Ensemble UCB ranks mean predicted improvement plus ensemble disagreement, and Gradient Information applies log-determinant design to reliability-weighted update gradients. RMWorld changes this last operator by adding the nonnegative task lower-confidence term in~\eqref{eq:set-objective}. Hyperparameters, secondary tie-breaking by candidate index, and the candidate order were frozen before formal seeds, while common random numbers ensure that a selector cannot benefit from an easier rollout realization. The comparison therefore isolates alternative evidence-selection rules under shared models, budgets, and optimizers.

We report medians and interquartile ranges because several endpoints are skewed. Paired differences use percentile bootstrap confidence intervals, exact win/tie/loss counts, and two-sided Wilcoxon signed-rank tests. Holm correction is applied over each prespecified family of baseline comparisons. A comparison is described as supported only when the adjusted test rejects and the paired direction favors RMWorld. Association regret, collision incidence, label yield, and representative trajectories remain secondary diagnostics unless a separate corrected test is reported. Across all tables, light-blue rows identify the primary evidence block, oracle or proposed configuration, frozen operating point, or decisive RMWorld-family contrast; boldface marks the key value, best value, or statistically resolved effect, and underlining marks a descriptive runner-up. These visual cues guide reading but do not replace confidence intervals or multiplicity-corrected tests.

\subsection{Paired Estimands and Leakage Controls}

The experimental unit is an entire frozen task and channel realization, not an individual grid link, time step, or UAV. In the 3GPP block, one trial fixes three independent random streams: the spatial shadow realization, the initial formula ensemble, and the task distribution. All selectors receive exactly those objects, the same four user-balanced warm links, the same candidate lattice, and the same 28 remaining labels. A fresh residual estimator is constructed for every method, preventing calibrated coefficients or queried links from leaking across a paired comparison. The $81\times81$ evaluation grid is never used as training data: only selected points on the nested $21\times21$ candidate lattice reveal the realized rate. Evaluation on all 26,244 user-location pairs occurs after fitting and therefore measures spatial generalization from 32 labels rather than interpolation over a fully observed map.

The realized 3GPP oracle uses smooth state-specific shadow fields produced from independent latent Gaussian fields and scaled by the standardized LOS and NLOS shadow terms. The estimated members do not observe those fields. Their fixed bias templates are listed in Table~\ref{tab:formula-ensemble}; small trial-wise jitter is applied to each template before selection begins. The templates straddle the true 60~m and 3.5~GHz configuration but do not include the hidden shadow realization. Each user-specific correction has an intercept plus a normalized $4\times4$ radial-basis design, uses unit residual observation variance in the acquisition covariance, and applies ridge penalties 1 (intercept) and 12 (other coefficients). Residual heads bootstrap only the queried log-rate ratios. The resampling indices are common across methods at every checkpoint, so their warm-start fits are identical and subsequent fitted differences originate from selected evidence rather than method-specific bootstrap noise.

\begin{table}[!t]
\centering
\caption{Primary formula-radio-WM templates before trial-wise jitter. Offsets $b_P$, $b_{\rm L}$, and $b_{\rm N}$ act on the LOS logit and LOS/NLOS path loss.}
\label{tab:formula-ensemble}
\resizebox{\columnwidth}{!}{%
\begin{tabular}{@{}l|r|r|r|r|r@{}}
\hline
Member & $h$ (m) & $f_c$ (GHz) & $b_P$ & $b_{\rm L}$ (dB) & $b_{\rm N}$ (dB) \\
\hline
\rowcolor{blue!8} \best{True formula} & \best{60} & \best{3.50} & 0 & 0 & 0 \\
Radio WM 1 & 42 & 3.20 & -0.65 & 1.50 & 3.00 \\
Radio WM 2 & 52 & 3.50 & -0.25 & 0.50 & 1.00 \\
Radio WM 3 & 72 & 3.80 & 0.20 & -0.50 & -1.00 \\
\hline
\end{tabular}}
\end{table}

For DeepMIMO, paired seeds fix the rate map, UAV starts, user demands, obstacles, policy initialization, and common branch perturbations. The multi-trial selectors share candidate count, retained count, high-fidelity audits, and local update machinery; radio-WM rollout calls are recorded rather than assumed equal. This distinction matters because a link-query budget is a data constraint, while a rollout count is a computational constraint. The low-cost Single-Trial radio WM violates the latter matching condition and is therefore labeled as a reference. The neural-feature study returns to exact link-budget matching and reuses the first 30 primary channel/task pairs, changing only the residual feature map and its representation seed.

\begin{table}[!t]
\centering
\caption{Frozen paired estimands and inferential scope. Formula improvements are comparator minus RMWorld; backlog differences are RMWorld minus comparator. $H$ is the Holm family size.}
\label{tab:estimands}
\resizebox{\columnwidth}{!}{%
\begin{tabular}{@{}l|c|l|c|l@{}}
\hline
Block & $n$ & Paired estimand & $H$ & Role \\
\hline
\rowcolor{blue!8} \best{3GPP formula} & \best{100} & wRMSE gain & \best{5} & \best{Primary} \\
DeepMIMO control & 30; 10/load & Backlog difference & 5; -- & External / boundary \\
Neural features & 30 & wRMSE gain & 5 & Sensitivity \\
Internal audit & 30/row & Backlog difference & 10 & Attribution \\
Stress tests & 30 & Median [IQR] & -- & Boundary \\
\hline
\end{tabular}}
\end{table}

Table~\ref{tab:estimands} records the sign convention and statistical family before presenting results. Bootstrap resampling operates on paired trial differences, so its interval targets the median within-instance effect rather than the difference between two unpaired marginal medians. Win/tie/loss counts expose sign stability that a single central estimate can conceal, while rank-biserial correlation reports paired effect direction. A result is not promoted when a nominal test and bootstrap interval disagree, when multiplicity correction removes significance, or when the comparator receives a materially different compute budget. This decision rule is responsible for the deliberately narrower DeepMIMO claim in Sec.~\ref{sec:deepmimo-results}.

\subsection{Primary 3GPP Formula-Radio-WM Evidence}
\label{sec:3gpp-results}

\begin{figure*}[!t]
  \centering
  \includegraphics[width=\textwidth]{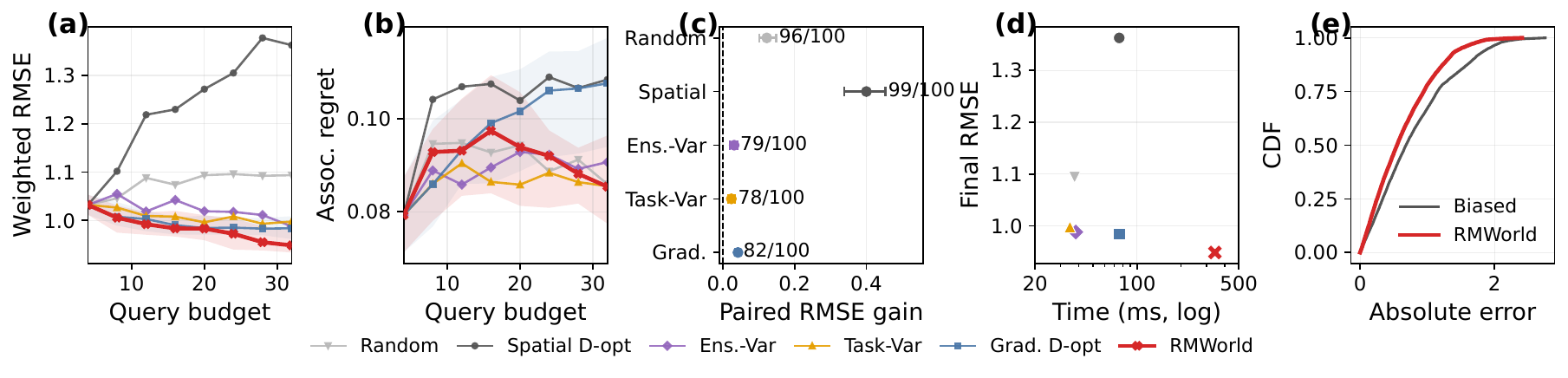}
  \caption{Five-panel primary result from 100 paired standardized 3GPP TR~36.777 UMa-AV analytic simulations. The panels report (a) task-weighted RMSE versus budget, (b) association regret versus budget, (c) paired final effects with bootstrap intervals and win counts, (d) the final error-cost frontier, and (e) the spatial absolute-error distribution in the representative mechanism trial. Curves show medians, and highlighted bands are 95\% bootstrap confidence intervals. All methods use 32 link labels including four warm-start labels.}
  \label{fig:3gpp-primary}
\end{figure*}

\begin{table}[!t]
\centering
\caption{Primary 3GPP formula-radio-WM evidence over 100 paired trials at 32 labels. The upper panel reports marginal endpoints; the lower panel reports comparator-minus-RMWorld wRMSE effects. Positive effects and RBC favor RMWorld. Tests use paired Wilcoxon signed ranks with Pratt zeros and Holm correction.}
\label{tab:3gpp-primary}
\scriptsize
\setlength{\tabcolsep}{2.7pt}
\resizebox{\columnwidth}{!}{%
\begin{tabular}{@{}l|c|c|r@{}}
\hline
Method & wRMSE [IQR] & Regret [IQR] & Time (ms) \\
\hline
Random & 1.094 [1.012,1.200] & 0.086 [0.066,0.122] & 37.3 \\
Spatial D-opt & 1.363 [1.204,1.569] & 0.108 [0.077,0.159] & 75.6 \\
Ens.-Var & 0.988 [0.940,1.050] & 0.091 [0.067,0.122] & 38.0 \\
Task-Var & 0.998 [0.923,1.050] & \second{0.086} [0.064,0.111] & \best{34.7} \\
Grad. D-opt & \second{0.984} [0.927,1.096] & 0.108 [0.075,0.133] & 76.0 \\
\rowcolor{blue!8} \best{RMWorld} & \best{0.949} [0.905,1.020] & \best{0.085} [0.064,0.110] & 342.8 \\
\rowcolor{gray!15} RMWorld gain vs. Task-Var & \best{4.9\%} & \best{0.001} & $+308.1$ \\
\hline
\end{tabular}}
\vspace{2pt}

\resizebox{\columnwidth}{!}{%
\begin{tabular}{@{}l|c|c|c|c@{}}
\hline
Comparator & Gain [95\% CI] & W/T/L & $p_{\rm Holm}$ & RBC \\
\hline
Random & 0.123 [0.102,0.149] & 96/0/4 & $<0.001$ & 0.977 \\
Spatial D-opt & 0.400 [0.338,0.453] & 99/0/1 & $<0.001$ & 0.999 \\
Ens.-Var & 0.032 [0.019,0.041] & 79/0/21 & $<0.001$ & 0.710 \\
Task-Var & 0.024 [0.015,0.035] & 78/0/22 & $<0.001$ & 0.634 \\
Grad. D-opt & 0.042 [0.034,0.049] & 82/0/18 & $<0.001$ & 0.867 \\
\hline
\end{tabular}}
\end{table}

At the fixed 32-label budget, RMWorld attains a median task-weighted formula RMSE of $0.949$~bit/s/Hz. Ensemble Variance, Task-Weighted Variance, and Gradient D-opt reach $0.988$, $0.998$, and $0.984$~bit/s/Hz. The corresponding paired improvements are $0.032$, $0.024$, and $0.042$~bit/s/Hz, and RMWorld wins 79, 78, and 82 of 100 trials. Every comparison remains significant after Holm correction with $p_{\rm Holm}<0.001$. The primary claim is therefore consistent reduction of task-weighted formula error, not an order-of-magnitude effect. Spatial D-opt is the weakest selector despite covering the candidate geometry broadly, and its $1.363$~bit/s/Hz median shows that parameter-space volume and rate-domain task error can diverge sharply under nonlinear LOS/NLOS mixing. Pointwise Ensemble Variance is much stronger, which confirms that uncertainty matters, but it does not account for how a label changes integrated task error. Task-Weighted Variance narrows the gap, yet RMWorld still wins 78 paired trials because covariance reduction distinguishes an informative feature from a merely important location. Gradient D-opt wins on some instances but loses to RMWorld in 82 trials because its determinant objective does not directly minimize integrated posterior rate variance. These contrasts connect the aggregate gain to the value-of-information mechanism rather than to any one spatial sampling heuristic.

The association endpoint provides a cautious decision-level check. RMWorld has the lowest descriptive median regret, $0.085$~bit/s/Hz, compared with $0.086$ for Task-Weighted Variance and Random and $0.108$ for Gradient D-opt. Intermediate regret is nonmonotone in Fig.~\ref{fig:3gpp-primary}, which is expected because user association changes discontinuously at an argmax boundary. We therefore do not infer that every RMSE improvement must reduce regret or claim a corrected significance result for this endpoint. The role of this metric is to confirm that the primary calibration gain is not obviously purchased by worse association.

\begin{figure}[!t]
  \centering
  \includegraphics[width=\columnwidth]{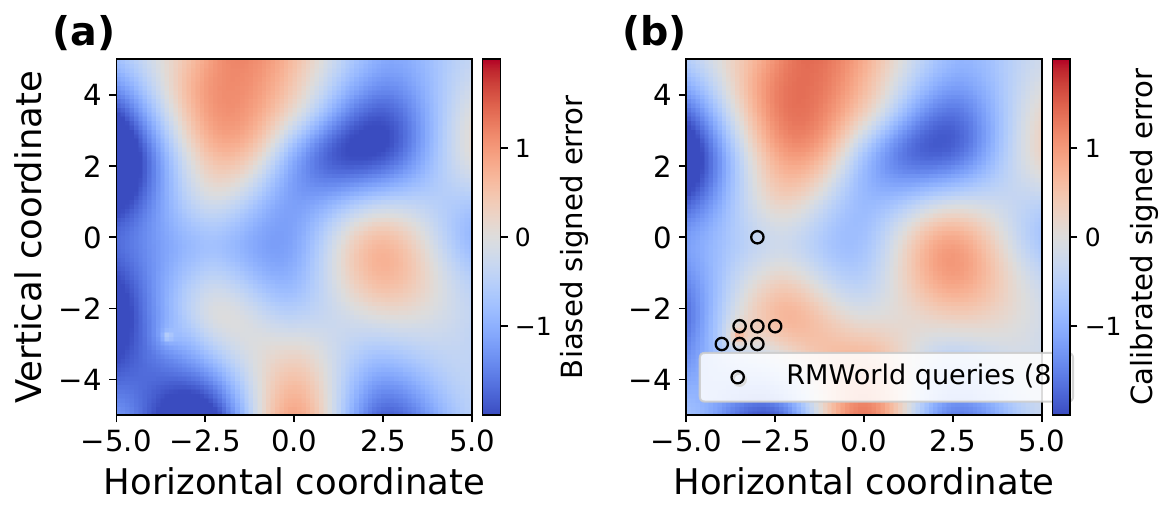}
  \caption{Spatial calibration mechanism in the median-RMWorld trial for the user with the largest task-weighted gain. Panel (a) shows the biased signed formula error, and panel (b) overlays RMWorld queries on the calibrated signed error. This standardized analytic realization explains the aggregate result but does not estimate a population effect.}
  \label{fig:3gpp-mechanism}
\end{figure}

Fig.~\ref{fig:3gpp-mechanism} traces the aggregate improvement to a spatial mechanism. The initial formula residual is structured rather than independent, so nearby links share both bias and posterior information. The task distribution concentrates on service corridors and user neighborhoods instead of the full grid. RMWorld selects links whose residual features couple those regions to large rate derivatives, then shifts the task-local error distribution leftward. This evidence is consistent with the variance-reduction derivation, although a single representative map is not used as a statistical claim. The same mechanism carries a clear offline price: in the deterministic rerun, RMWorld takes a median $342.8$~ms per trial, compared with $34.7$~ms for Task-Weighted Variance and $76.0$~ms for Gradient D-opt. The ratios are approximately $9.9\times$ and $4.5\times$, respectively, because integrated quadratic forms are reevaluated for all candidate links after each sequential covariance update. This computation changes training-time evidence allocation but does not add a channel-acquisition module to deployment inference.

\begin{figure}[!t]
  \centering
  \includegraphics[width=\columnwidth]{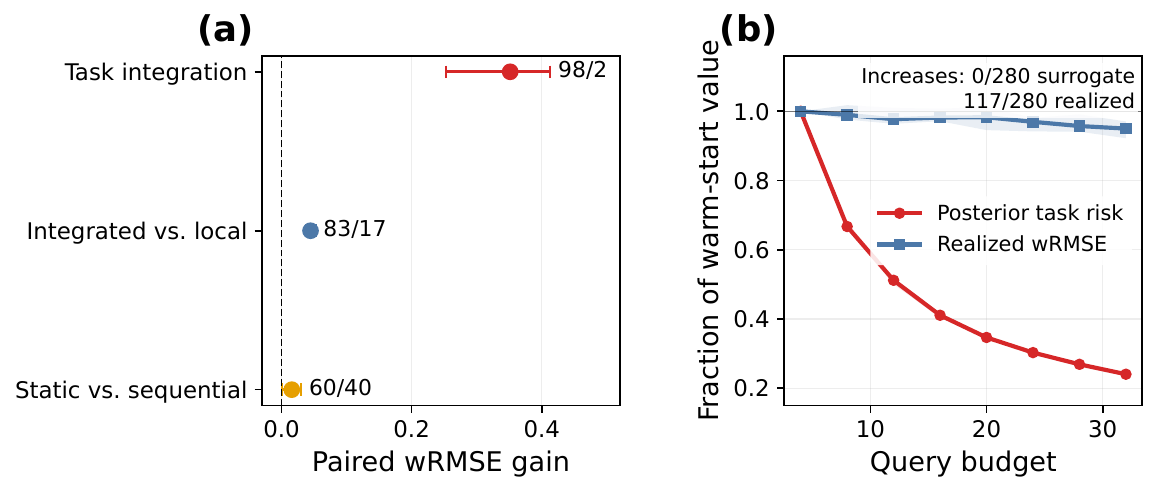}
  \caption{Attribution and theory--endpoint boundary. Panel (a) reports paired 32-label wRMSE effects with 95\% bootstrap intervals and win/loss counts. Panel (b) normalizes the posterior surrogate and realized wRMSE to each trial's warm-start value over 40 instrumented trials.}
  \label{fig:3gpp-audit}
\end{figure}

\begin{table}[!t]
\centering
\caption{Controlled 3GPP attribution and theory--endpoint audit. Positive gain favors the first named method. Holm correction covers three contrasts.}
\label{tab:3gpp-mechanism-audit}
\scriptsize
\setlength{\tabcolsep}{2.5pt}
\resizebox{\columnwidth}{!}{%
\begin{tabular}{@{}l|c|c|c@{}}
\hline
Arm & Task scope & Seq. & wRMSE \\
\hline
Task-Var & Local & No & 0.998 \\
Seq. Task-Var & Local & Yes & 0.983 \\
A-opt ($\bm H=\bm I$) & None & Yes & 1.306 \\
\rowcolor{blue!8} \best{RMWorld} & \best{Integrated} & Yes & \best{0.949} \\
\hline
Contrast & Gain [95\% CI] & W/T/L & $p_{\rm H}$ \\
\hline
RMWorld vs. A-opt & \best{0.352 [0.253,0.413]} & 98/0/2 & \best{$<0.001$} \\
RMWorld vs. Seq. Task-Var & \best{0.045 [0.037,0.052]} & 83/0/17 & \best{$<0.001$} \\
Static vs. sequential Task-Var & 0.016 [0.0001,0.030] & 60/0/40 & 0.021 \\
\hline
Quantity & Final reduction & Increases & Final worse \\
\hline
\rowcolor{blue!8} \best{Posterior surrogate} & \best{75.9\%} & \best{0/280} & -- \\
Realized wRMSE & 5.0\% & 117/280 & 9/40 \\
\hline
\end{tabular}}
\end{table}

The attribution audit separates task integration from within-batch re-scoring under the same paired instances. Replacing the identity Gram matrix by the joint integrated task Gram yields a median $0.352$~bit/s/Hz gain and wins 98 of 100 trials. Against sequential pointwise Task-Var, integration contributes $0.045$ and wins 83 trials. The sequential-only contrast also illustrates why the paired estimand matters: its marginal medians are $0.998$ (static) and $0.983$ (sequential), yet within-instance differences favor static Task-Var in 60 trials, with a paired median $0.016$ and $p_{\rm Holm}=0.021$. Thus task integration produces the robust measured gap; sequential re-scoring alone is instance-dependent and does not explain RMWorld's advantage. Fig.~\ref{fig:3gpp-audit}(b) also prevents the posterior identity from being mistaken for a realized-error guarantee. Over 280 acquisition transitions, the instrumented $\sum_u\tr(\bm H_u\bm\Sigma_u)$ trajectory never increases and falls by a median $75.9\%$, whereas realized wRMSE falls by only $5.0\%$ at the endpoint, increases in 117 transitions, and finishes above its warm-start value in 9 of 40 trials. This gap is compatible with Theorem~\ref{thm:voi}: the theorem controls posterior variance inside the local residual model, not omitted shadow structure, refit-induced derivative changes, or realized squared bias.

\subsection{End-to-End DeepMIMO Control}
\label{sec:deepmimo-results}

\begin{figure}[!t]
  \centering
  \includegraphics[width=\columnwidth]{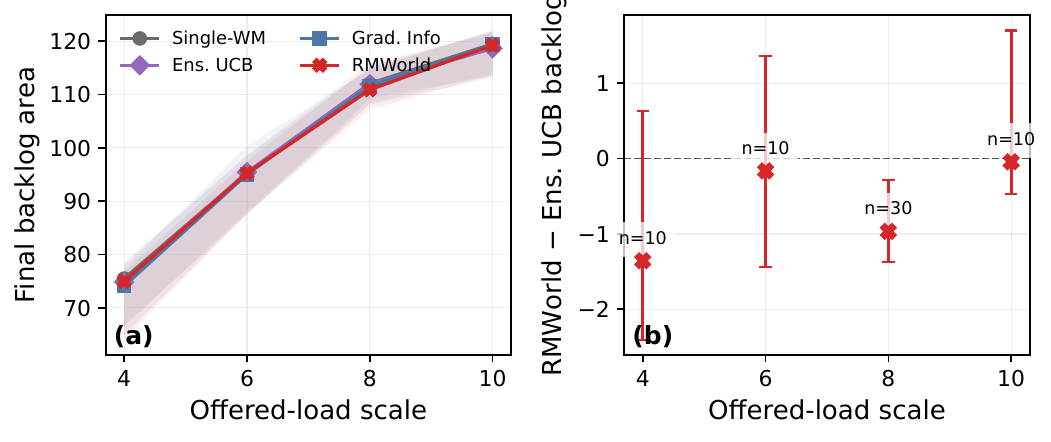}
  \caption{Same-protocol DeepMIMO load boundary. Panel (a) reports marginal median backlog with interquartile bands; panel (b) reports the within-seed median RMWorld-minus-Ensemble-UCB effect and 95\% bootstrap interval. Negative effects favor RMWorld. Load 8 is the prespecified 30-seed comparison; loads 4, 6, and 10 are 10-seed exploratory probes outside the Holm family.}
  \label{fig:deepmimo-evidence}
\end{figure}

\begin{table}[!t]
\centering
\caption{Paired high-load DeepMIMO control evidence over 30 seeds. Backlog is median [IQR]. The lower panel reports RMWorld-minus-comparator paired effects and Holm-adjusted Wilcoxon tests. Single-Trial radio WM is an unequal-budget, low-cost reference.}
\label{tab:deepmimo-main}
\scriptsize
\setlength{\tabcolsep}{2.7pt}
\resizebox{\columnwidth}{!}{%
\begin{tabular}{@{}l|c|r|r@{}}
\hline
Method & Backlog [IQR] & Radio-WM calls & Yield (\%) \\
\hline
Single-Trial radio WM & 112.055 [108.058,114.903] & \best{24} & 0.0 \\
Random MT & 112.529 [108.119,115.932] & \second{192} & 41.7 \\
Action Entropy & \second{111.129} [108.269,114.805] & \second{192} & 41.7 \\
Ensemble UCB & 111.926 [108.318,116.001] & \second{192} & 33.3 \\
Gradient Info & 111.583 [109.245,114.654] & \second{192} & \second{54.2} \\
\rowcolor{blue!8} \best{RMWorld} & \best{110.896} [107.185,115.296] & 264 & \best{62.5} \\
\rowcolor{gray!15} Gain vs. Ens. UCB & \best{$-0.967$ paired} & $+72$ & \best{$+29.2$ pp} \\
\hline
\end{tabular}}
\vspace{2pt}

\resizebox{\columnwidth}{!}{%
\begin{tabular}{@{}l|c|c|c@{}}
\hline
Comparator & RMWorld $\Delta$ [95\% CI] & W/T/L & $p_{\rm Holm}$ \\
\hline
Single-Trial radio WM & -0.677 [-0.997,-0.063] & 20/0/10 & 0.028 \\
Random MT & -0.503 [-1.337,0.339] & 19/0/11 & 0.046 \\
Action Entropy & -0.679 [-1.358,0.267] & 19/0/11 & 0.095 \\
\rowcolor{red!6} Ensemble UCB & \best{-0.967 [-1.372,-0.290]} & \best{23/0/7} & \best{0.010} \\
Gradient Info & -0.171 [-0.721,0.455] & 17/0/13 & 0.339 \\
\hline
\end{tabular}}
\end{table}

RMWorld has the lowest descriptive median backlog, $110.896$, in the high-load control task. The strongest supported equal-budget external comparison is Ensemble UCB: the paired median difference is $-0.967$, its bootstrap interval remains below zero, RMWorld wins 23 of 30 seeds, and the Holm-adjusted test gives $p=0.010$. This comparison matters because Ensemble UCB receives the same branch pool and high-fidelity audit count but chooses optimistic predicted improvement rather than task-credible gradient evidence. The result supports end-to-end transfer of the selection principle to a channel map that is outside the analytic formula family. The remaining contrasts prevent a universal superiority statement: the nominal test against Random Multi-Trial gives $p=0.046$, but its bootstrap interval crosses zero, so the evidence is fragile across summaries. Action Entropy and Gradient Information are unresolved with adjusted $p$ values of $0.095$ and $0.339$. Single-Trial radio WM gives a favorable paired result but uses only 24 radio-WM calls and is not an equal-budget selector. Accordingly, the paper claims a supported advantage over Ensemble UCB and reports the other comparisons without collapsing them into an ``all baselines'' claim.

Fig.~\ref{fig:deepmimo-evidence} marks the operating boundary rather than implying monotone gain. The exploratory RMWorld-minus-UCB effects are $-1.356$ [$-2.405,0.631$], $-0.164$ [$-1.445,1.364$], and $-0.043$ [$-0.475,1.700$] at loads 4, 6, and 10; only the prespecified load-8 interval excludes zero. Loads 4 and 10 also expose why the paired estimand is indispensable: their marginal RMWorld/UCB medians are $74.954/74.830$ and $119.250/118.666$, respectively, although the median within-seed differences are negative. This apparent discrepancy is not contradictory because medians are not additive, and it rules out a universal or load-monotone advantage by localizing the supported result to the severe load-8 regime. Table~\ref{tab:deepmimo-main} then exposes the corresponding cost-quality tradeoff. RMWorld uses 264 radio-WM rollout calls versus 192 for each multi-trial selector, and its median runtime is $33.15$~s versus $26.4$--$26.8$~s, amounting to $37.5\%$ more calls and about $25\%$ more wall time than Ensemble UCB. Its confirmed-label yield is $62.5\%$, 29.2 percentage points above UCB and 8.3 points above Gradient Information. Yield remains descriptive because it is an intermediate mechanism, not the prespecified endpoint.

\subsection{Mechanism Audits and Stress Boundaries}
\label{sec:audits}

The selector-controlled DeepMIMO experiment is the appropriate external comparison because all multi-trial methods share the same update machinery. Additional experiments diagnose paired removals within the RMWorld family, candidate-pool sensitivity, and a longer training horizon; they are not pooled into the main ranking. Their purpose is to expose when the proposed ingredients matter and where extra computation stops improving the endpoint. The paired removal audit in Table~\ref{tab:internal-audit} therefore asks a narrower question: Basic Gradient Information removes update alignment and conflict projection, while Gradient Information with Update Alignment removes only conflict projection. The reported difference is RMWorld minus the ablation, so a negative value favors the full update. Holm correction covers all scene-by-ablation comparisons in this audit family. Only the severe O1 altitude shift resolves the contribution beyond the adjusted threshold.

\begin{table}[!t]
\centering
\caption{Paired internal RMWorld audit over 30 seeds per row. Removed A+P denotes update alignment and conflict projection; P denotes projection only. Negative differences favor the full update, and $p_{\rm H}$ corrects all ten contrasts.}
\label{tab:internal-audit}
\scriptsize
\setlength{\tabcolsep}{2.0pt}
\resizebox{\columnwidth}{!}{%
\begin{tabular}{@{}l|c|c|c|c|c@{}}
\hline
Scene & Removed & Median $\Delta$ [95\% CI] & W/T/L & $p_{\rm H}$ & RBC \\
\hline
O1/0.0 m & A+P & -0.012 [-0.056,0.063] & 16/0/14 & 1.0000 & 0.006 \\
O1/0.0 m & P & -0.043 [-0.087,0.002] & 19/0/11 & 1.0000 & -0.260 \\
O1/0.8 m & A+P & -0.018 [-0.061,0.025] & 17/0/13 & 1.0000 & -0.028 \\
O1/0.8 m & P & -0.016 [-0.103,0.004] & 18/0/12 & 0.9546 & -0.277 \\
O1/1.6 m & A+P & -0.029 [-0.094,-0.007] & 21/0/9 & 0.6619 & -0.260 \\
O1/1.6 m & P & \second{-0.066 [-0.122,-0.013]} & 23/0/7 & \second{0.0602} & -0.518 \\
O1/2.4 m & A+P & -0.027 [-0.083,-0.002] & 21/0/9 & 1.0000 & -0.131 \\
\rowcolor{blue!8} \best{O1/2.4 m} & \best{P} & \best{-0.070 [-0.108,-0.025]} & \best{24/0/6} & \best{0.0109} & \best{-0.622} \\
ASU/cross & A+P & 0.204 [-1.708,2.366] & 13/0/17 & 1.0000 & 0.015 \\
ASU/cross & P & 0.136 [-0.898,1.120] & 13/0/17 & 1.0000 & -0.002 \\
\hline
\end{tabular}}
\end{table}

At the 2.4~m O1 shift, adding conflict projection to the aligned update produces a paired median difference of $-0.070$, a bootstrap interval of $[-0.108,-0.025]$, 24 wins in 30 seeds, and $p_{\rm Holm}=0.0109$. The milder O1 shifts show favorable paired directions but do not survive family-wise correction. The ASU cross-scene rows do not favor RMWorld and have wide intervals, indicating that the same local update geometry does not transfer automatically across map families. This pattern supports the projection as a severe-mismatch stabilizer, not a universal source of backlog reduction. A separately frozen branch-budget sweep gives an additional non-result: increasing $(B,K)$ from $(3,1)$ to $(7,5)$ changes median backlog only from $38.840$ to $38.910$, while radio-WM calls double from 96 to 192. The default $(5,2)$ is therefore a moderate-cost operating point, not an empirically unique optimum.

\subsection{Representation Sensitivity and Boundary Tests}
\label{sec:sensitivity}

\begin{figure}[!t]
  \centering
  \includegraphics[width=\columnwidth]{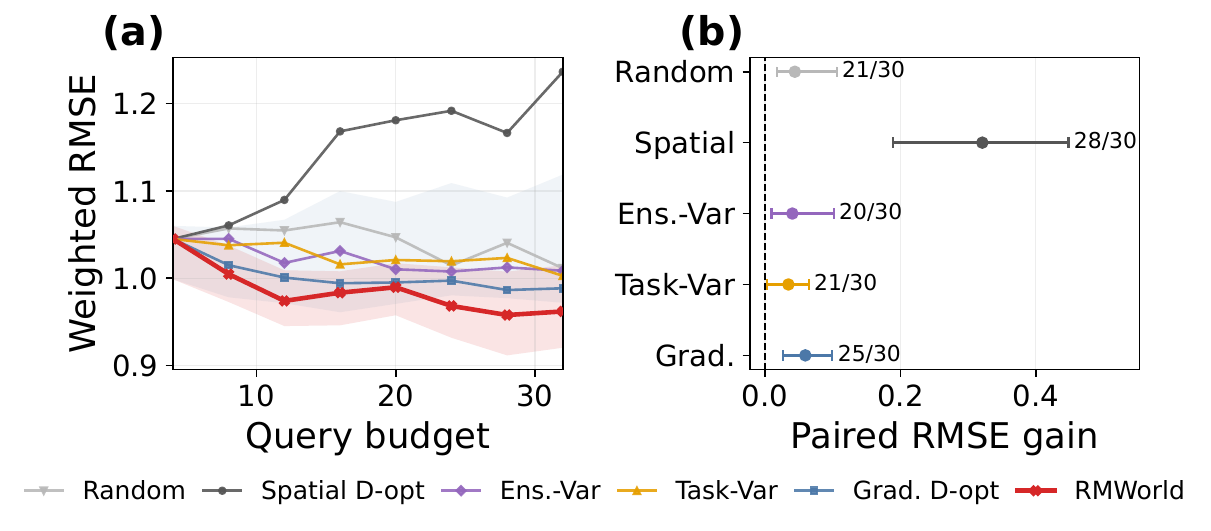}
  \caption{Neural-feature radio-WM sensitivity over 30 paired trials. The residual model uses a frozen two-layer tanh random-feature encoder with a query-trained ridge output head. Panel (a) reports RMSE versus query budget, and panel (b) reports paired final effects. This experiment does not constitute fully end-to-end neural-network training.}
  \label{fig:neural-sensitivity}
\end{figure}

\begin{table}[!t]
\centering
\caption{Neural-feature radio-WM sensitivity over 30 paired trials at 32 labels. The upper panel reports marginal endpoints; the lower panel reports comparator-minus-RMWorld wRMSE effects. Positive effects and RBC favor RMWorld; $p_{\rm H}$ is Holm adjusted.}
\label{tab:neural-results}
\scriptsize
\setlength{\tabcolsep}{2.7pt}
\resizebox{\columnwidth}{!}{%
\begin{tabular}{@{}l|c|c|r@{}}
\hline
Method & wRMSE [IQR] & Regret [IQR] & Time (ms) \\
\hline
Random & 1.012 [0.934,1.133] & 0.104 [0.083,0.132] & \best{116.9} \\
Spatial D-opt & 1.236 [1.141,1.433] & 0.125 [0.101,0.156] & 157.0 \\
Ens.-Var & 1.009 [0.954,1.063] & 0.128 [0.086,0.137] & \second{123.0} \\
Task-Var & 1.003 [0.945,1.062] & \second{0.102} [0.072,0.119] & 124.7 \\
Grad. D-opt & \second{0.989} [0.951,1.130] & 0.109 [0.079,0.134] & 383.9 \\
\rowcolor{blue!8} \best{RMWorld} & \best{0.962} [0.906,1.052] & \best{0.095} [0.068,0.116] & 446.6 \\
\rowcolor{gray!15} RMWorld gain vs. Grad. D-opt & \best{2.7\%} & \best{0.014} & $+62.7$ \\
\hline
\end{tabular}}
\vspace{2pt}

\resizebox{\columnwidth}{!}{%
\begin{tabular}{@{}l|c|c|c|c@{}}
\hline
Comparator & Gain [95\% CI] & W/T/L & $p_{\rm H}$ & RBC \\
\hline
Random & 0.044 [0.018,0.106] & 21/0/9 & 0.012 & 0.587 \\
Spatial D-opt & 0.321 [0.190,0.448] & 28/0/2 & $<0.001$ & 0.987 \\
Ens.-Var & 0.040 [0.010,0.102] & 20/0/10 & 0.037 & 0.488 \\
Task-Var & 0.035 [0.003,0.065] & 21/0/9 & 0.037 & 0.480 \\
Grad. D-opt & 0.060 [0.026,0.100] & 25/0/5 & 0.004 & 0.665 \\
\hline
\end{tabular}}
\end{table}

The neural-feature study changes the residual representation while keeping the channel, task, and budget instances paired with the first 30 formula trials. RMWorld reaches a median RMSE of $0.962$~bit/s/Hz, compared with $1.009$ for Ensemble Variance, $1.003$ for Task-Weighted Variance, and $0.989$ for Gradient D-opt. It wins 20, 21, and 25 of 30 paired comparisons, respectively; all Holm-adjusted tests remain below $0.05$ (range $0.004$--$0.037$). Paired formula and neural RMWorld errors have Spearman correlation $0.79$, while the median neural-minus-formula shift is $-0.011$~bit/s/Hz. These are representation diagnostics, not evidence of end-to-end deep-model generalization. The association diagnostic is weaker than the primary RMSE result: RMWorld and Task-Weighted Variance have medians of $0.095$ and $0.102$~bit/s/Hz, and the table does not attach a corrected decision-level test. This descriptive difference is consistent with, but does not resolve, behavior at discontinuous association boundaries. It also shows why the paper does not treat association regret as a guaranteed monotone consequence of integrated variance reduction. The primary conclusion remains about task-weighted radio-WM error.

The 30-seed long-horizon stress states another boundary rather than hiding it. After 12 policy updates, RMWorld has median backlog $38.146$, while Single-Trial radio WM, Basic Gradient Information, and Gradient Information with alignment obtain $38.448$, $37.879$, and $38.298$. Their interquartile ranges overlap substantially, so the experiment does not support long-horizon superiority. This outcome suggests that the available radio-WM signal saturates before additional counterfactual updates create stable separation; more training is not a substitute for better evidence quality here. Across the three evidence blocks, offline computation is therefore the consistent tradeoff. Channel-level RMWorld is slower because it recomputes integrated variance reductions, and policy-level RMWorld uses additional counterfactual rollouts plus fixed-batch validation. The method is query-budget matched in its primary comparisons but is not uniformly training-compute matched. Its deployment cost remains one policy forward pass because acquisition, formula fitting, branch selection, and line search stop after training. This distinction prevents ``sample efficient'' from being used as a synonym for ``computationally cheap.''

\subsection{Threats to Validity and Reproducibility}

\begin{table}[!t]
\centering
\caption{Validity matrix. Each mitigation is enforced in the reported protocol; the last column states the claim that remains outside the evidence.}
\label{tab:validity}
\scriptsize
\setlength{\tabcolsep}{2.5pt}
\resizebox{\columnwidth}{!}{%
\begin{tabular}{@{}l|l|l@{}}
\hline
Threat & Protocol mitigation & Remaining scope \\
\hline
Oracle realism & 3GPP formula + held-out DeepMIMO & No field measurement \\
WM misspecification & Biased physics + hidden shadow + neural features & No dynamic blockage \\
\rowcolor{blue!8} \best{Causal attribution} & \best{One level varied per primary block} & \best{Joint synergy not identified} \\
Multiplicity & Frozen families, paired CI/W/T/L, Holm tests & Load scan exploratory \\
Compute fairness & Queries fixed; calls and time reported & Training FLOPs unequal \\
Representation & Formula and random-feature residuals & No end-to-end neural WM \\
\hline
\end{tabular}}
\end{table}

The evidence supports a scoped causal claim rather than universal radio intelligence. The standardized 3GPP block makes the true and estimated propagation formulas explicit and introduces hidden correlated shadow structure that the learner cannot observe directly~\cite{3GPPTR36777}. This transparency isolates evidence allocation from an opaque oracle, but the rates remain simulated rather than measured. Freezing the channel field, task distribution, warm start, formula ensemble, and query budget within every paired trial removes the main nuisance variation between selectors. Separating channel calibration from policy selection across the two primary blocks further prevents a compound improvement from being assigned to the wrong level. The exact identity in Theorem~\ref{thm:voi}, however, is conditional on the local linear-Gaussian residual and therefore does not guarantee monotone realized error under omitted physics. Discontinuous association and collision outcomes likewise remain diagnostics unless their own corrected tests are reported. Table~\ref{tab:validity} records these boundaries beside the mitigation that makes each retained claim interpretable.

External validity is tested by moving policy learning to a held-out DeepMIMO ray-tracing slice whose rates are outside the analytic formula family~\cite{alkhateeb2019deepmimo}, but one outdoor map cannot represent all blockage, mobility, or fleet regimes. The load scan and cross-scene audit are therefore used to locate failure boundaries, not to enlarge the confirmatory family after observing results. The neural-feature study changes the residual basis under matched channel and task instances, yet its frozen random encoder does not establish end-to-end neural-radio-WM generalization. Computational validity is treated separately from query efficiency: label and audit budgets are matched where claimed, while rollout calls, runtime, and unequal-cost references remain visible. Reproducibility follows the same separation of concerns. Every headline statistic is regenerated from stored per-trial arrays through deterministic analysis scripts, the common-random-number streams and tie-breaking order are frozen, and the artifact manifest links each table or figure to its aggregate source. These controls make the reported paired estimands auditable, although independent measured-channel replication remains necessary before deployment claims can be made.

\section{Conclusion}
\label{sec:conclusion}

RMWorld has formulated multi-UAV learning with imperfect radio world models as two coupled evidence-allocation problems and has linked task-integrated radio-WM calibration to update-aligned counterfactual selection. It has established an exact local channel variance-reduction identity, a submodular greedy guarantee for branch selection, and a scoped first-order result for conflict projection. The resulting framework improves the consistency of task-weighted formula-based radio-WM calibration and provides a supported severe-load gain over Ensemble UCB; load probes delimit that gain as nonuniversal, while cost audits expose the offline computation and unresolved comparisons. Future work will test measured channels, dynamic blockage, larger fleets, and fully trained neural radio world models.

\bibliographystyle{IEEEtran}
\bibliography{ref}

\end{document}